\documentclass[11pt, reco]{article}
\usepackage{amssymb,amscd,amsmath,amsthm}
\usepackage[colorinlistoftodos]{todonotes}
\usepackage[colorlinks=true, allcolors=blue]{hyperref}
\usepackage{lmodern}
\usepackage[english]{babel}
\usepackage{color}
\usepackage{graphicx}
\graphicspath{ {./images/} }
\newtheorem{theorem}{Theorem}[section]

\newtheorem{corollary}[theorem]{Corollary}
\newtheorem{lemma}[theorem]{Lemma}
\newtheorem{definition}[theorem]{Definition}
\newtheorem{example}[theorem]{Example}
\newtheorem{remark}[theorem]{Remark}

\def\id{{\bf 1}\!\!{\rm I}}

\def\tr{{\rm Tr}}
\def\L{\Lambda}
\def\G{\Gamma}
\def\ca{{\mathcal A}}

\def\cm{{\mathcal M}}
\def\ch{{\mathcal H}}
\def\ck{{\mathcal K}}
\def\r{\rho}
\def\f{\varphi}
\def\bz{{\mathbb Z}}
\def\bp{{\mathbb P}}

\begin{document}

\begin{center}
{\Large {\bf Open Quantum Random Walks and Quantum Markov
chains on Trees II: The recurrence }}\\[1cm]

{\sc Farrukh Mukhamedov$^*$} \\[2mm]
 Department of Mathematical Sciences, College of Science, \\
 United Arab Emirates University 15551, Al-Ain,\\ United
Arab Emirates\, and \\
Institute of Mathematics named after V.I.Romanovski, 4,\\
University str., 100125, Tashkent, Uzbekistan\\
e-mail: {\tt far75m@gmail.com; farrukh.m@uaeu.ac.ae}\\[1cm]

{\sc Abdessatar Souissi$^*$} \\[2mm]
Department of Accounting, College of Business Management\\
Qassim University, Ar Rass, Saudi Arabia \, and \\
Preparatory institute for scientific and technical studies,\\
 Carthage University, Amilcar 1054, Tunisia\\
  e-mail: {{\tt a.souaissi@qu.edu.sa}; {\tt
 abdessattar.souissi@ipest.rnu.tn}}\\[1cm]

 {\sc Tarek Hamdi} \\[2mm]
Department of Management Information Systems, College of Business Management\\
Qassim University, Ar Rass, Saudi Arabia \, and \\
Laboratoire d'Analyse Math\'ematiques et applications LR11ES11 \\
Universit\'e de Tunis El-Manar, Tunisia\\
  e-mail: {{\tt t.hamdi@qu.edu.sa}}\\[1cm]

   {\sc Amen Allah Andolsi} \\[2mm]
  Nuclear Physics and High Energy Physics Research Unit, Faculty of Sciences of Tunis,  University of Tunis El Manar,  2092 Tunis, Tunisia\\
  e-mail: {{\tt amenallah.andolsi@fst.utm.tn}}\\[1cm]

\end{center}

$^*$ Corresponding authors

\small
\begin{center}
{\bf Abstract}\\
\end{center}
In the present paper, we construct QMC (Quantum Markov Chains)
associated with Open Quantum Random Walks such that the transition
operator of the chain is defined by OQRW and the restriction of QMC
to the commutative subalgebra coincides with the distribution
of OQRW. Furthermore, we first propose a new construction of QMC on trees, which is an extension of QMC considered in Ref. \cite{AOM}. Using such a construction, we are able to construct QMCs on tress associated with OQRW. Our investigation leads to the detection of the phase transition phenomena within the proposed scheme. This kind of phenomena appears first time in this direction. Moreover, mean entropies of QMCs are calculated.


\normalsize

\section{Introduction}

Motivated largely by the prospect of superefficient algorithms, the theory of quantum Markov
chains (QMC), especially in the guise of quantum walks, has generated a huge number of works,
including many discoveries of fundamental importance \cite{Ke,Kum,nielsen,portugal}.
In \cite{Fing} a
novel approach  has been proposed to investigate quantum cryptography problems by means of QMC, where quantum effects are entirely encoded into super-operators labelling transitions, and
the nodes of its transition graph carry only classical information and thus they are discrete. In \cite{AW87,DK19} QMC have been applied  to the investigations of so-called "open quantum random walks" (OQRW) \cite{attal,carbone,carbone2,konno,cfrr}.
OQRW are related to the study of asymptotic behavior of trace-preserving completely
positive maps, which belong to fundamental topics of quantum information theory (see, for instance
\cite{burgarth2,petulante,novotny}). These quantum walks are possible noncommutative generalizations of classical
Markov chains and have applications in quantum computing, quantum optics \cite{KKS19,MSK13}. We refer the reader to \cite{SP19}
for a recent survey on the subject.

Recently, in \cite{MSH22} we first have proposed a new construction of QMC on trees, which is an extension of QMC considered in \cite{AOM}. Using such a construction, QMCs are defined on tress associated with OQRW. The investigation led to the detection of the phase transition phenomena within the proposed scheme. Such kind of phenomena appeared for the first time in this direction. In the present paper, we continue the proposed investigation to discuss the recurrence problem for the associated QMC. In one dimensional setting the recurrence problem had been paid attention by many authors (see for example \cite{accardi1,accardi2,DK19,GLV20,cfrr}).

For the sake of clarity, let us recall some necessary information about OQRW.
 Let $\mathcal{K}$ denote a separable Hilbert space and let
$\{|i\rangle\}_{i\in \L}$ be its orthonormal basis indexed by the
vertices of some graph $\L$ (here the set $\L$ of vertices might
be finite or countable). Let $\mathcal{H}$ be another Hilbert space,
which will describe the degrees of freedom given at each point of
$\L$. Then we will consider the space
$\mathcal{H}\otimes\mathcal{K}$. For each pair $i,j$ one associates
a bounded linear operator $B_{j}^i$ on $\mathcal{H}$. This operator
describes the effect of passing from $|j\rangle$ to $|i\rangle$. We
will assume that for each $j$, one has
\begin{equation}\label{Bij}
\sum_i B_{j}^{i*}B_{j}^i=\id,
\end{equation} where, if infinite,
such series is strongly convergent. This constraint means: the sum
of all the effects leaving site $j$ is $\id$. The operators $B^i_j$
act on $\mathcal{H}$ only, we dilate them as operators on
$\mathcal{H}\otimes\mathcal{K}$ by putting
$$
M^i_j=B^i_j\otimes \vert i\rangle\langle j\vert\,.
$$
The operator $M^i_j$ encodes exactly the idea that while passing
from $\vert j\rangle$ to $\vert i\rangle$ on the lattice, the effect
is the operator $B^i_j$ on $\mathcal{H}$.

According to \cite{attal} one has
\begin{equation}\label{sumMij=1}
\sum_{i,j} {M^i_j}^* M^i_j=\id.
\end{equation}

Therefore, the operators $(M^i_j)_{i,j}$ define a completely
positive mapping
\begin{equation}\label{MM}
\cm(\r)=\sum_i\sum_j M^i_j\,\r\, {M^i_j}^*
\end{equation}
on $\ch\otimes\ck$.

In what follows, we consider density matrices on
$\mathcal{H}\otimes\mathcal{K}$ which take the form
\begin{equation}\label{rr}
\rho=\sum_i\rho_i\otimes |i\rangle\langle i|,
\end{equation}
 assuming that $\sum_i\tr(\rho_i)=1$.

For a given initial state of such form, the \textit{Open Quantum
Random Walk (OQRW)} is defined by the mapping $\cm$, which has the
following form
\begin{equation}\label{MM1}
\cm(\rho)=\sum_i\Big(\sum_j B_{j}^i\rho_j B_{j}^{i*}\Big)\otimes
|i\rangle\langle i|.
\end{equation}

By means of the map $\cm$ one defines a family of classical random
process on $\O=\L^{\bz_+}$. Namely, for any density operator $\r$ on
$\mathcal{H}\otimes\mathcal{K}$ (see \eqref{rr}) the probability
distribution is defined by
\begin{equation}\label{Prb}
\bp_\r(i_0,i_1,\dots,i_n)=\tr(B^{i_n}_{i_{n-1}}\cdots
B^{i_2}_{i_{1}}B^{i_1}_{i_{0}}\rho_{i_0}B^{i_1*}_{i_{0}}B^{i_2*}_{i_{1}}\cdots
B^{i_n*}_{i_{n-1}}).
\end{equation}
We point out that this distribution is not a Markov measure
\cite{BBP}.

On the other hand, it is well-known \cite{AW87,norris} that to each
classical random walk one can associate a certain Markov chain and
some properties of the walk can be explored by the constructed
chain. Recently, in \cite{DK19,DM19}, we have found a quantum Markov
chain \footnote{We note that a Quantum Markov Chain is a
quantum generalization of a Classical Markov Chain where the state
space is a Hilbert space, and the transition probability matrix of a
Markov chain is replaced by a transition amplitude matrix, which
describes the mathematical formalism of the discrete time evolution
of open quantum systems, see
\cite{Ac,[AcFr80]} for more details.}
$\f$ on the
algebra $\ca=\otimes_{i\in\bz_+}\ca_i$, where $\ca_i$ is isomorphic to
$B(\ch)\otimes B(\ck)$, $i\in\bz_+$, such that the transition operator
$P$ equals to the mapping
$\cm^*$\footnote{The dual of $\cm$ is defined by the equality
$\tr(\cm(\r)x)=\tr(\r\cm^*(x))$ for all density operators $\rho$ and
observables $x$.} and the restriction of $\f$ to the commutative
subalgebra of $\ca$ coincides with the distribution $\bp_\r$, i.e.
\begin{equation}\label{I1}
\f\big((\id\otimes|i_0><i_0|)\otimes\cdots\otimes(\id\otimes|i_n><i_n|)\big)=\bp_\r(i_0,i_1,\dots,i_n).
\end{equation}
Hence, this result allows us  to
interpret the distribution $\bp_\r$ as a QMC, and to study further
properties of $\bp_\r$.

In \cite{MSH22}, we have initiated to look at the probability distribution \eqref{Prb} as a Markov field over the Cayley tree $\G^k$ 
Roughly speaking, $(i_0,i_1,\dots,i_n)$ is considered as a configuration on $\Omega=\L^{\G^k}$. Such kind of consideration allows us to investigated a phase transition phenomena associated for OQRW within QMC scheme \cite{MBS161,MBSG20}.

We stress that, in physics, a spacial classes of QMC, called "Matrix Product States" (MPS) and more generally "Tensor Network States" \cite{CV,Or} were used to investigate
quantum phase transitions for several lattice models.
This method uses the density matrix renormalization group (DMRG) algorithm which opened a new way of performing
the renormalization  procedure in 1D systems and gave extraordinary precise results. This is done by keeping the states of subsystems which
are relevant to describe the whole wave-function, and not those that minimize the energy on
that subsystems \cite{[RoOs96]}.

In this paper, we propose to investigate the recurrence problem for QMC on trees, and apply it to the QMC associated with OQRW on trees. Notice that the mentioned problem has been investigated for discrete-time nearest-neighbor open quantum random walks on the integer line in \cite{CGL17}. However, in the present work, we focus on the recurrence problem associated with QMC, while in \cite{CGL17,GLV20,JL21} the recurrence has
been teated with respect to the probability distribution \eqref{Prb}.

\section{Preliminaries}\label{sec_prel}

Let $\Gamma^k_{+} = (V,E)$ be the semi-infinite Cayley tree of order $k$  with root  $o$. The Cayley tree of order $k$ is characterized by being a tree for which every vertex has exactly
 $k+1$ nearest-neighbors. Recall that, two vertices $x$ and $y$ are  {\it nearest neighbors} (denoted $x\sim y$ ) if they are joined through an edge (i.e. $<x,y>\in E$). A   list $ x\sim x_1\sim \dots \sim x_{d-1}\sim y$ of vertices is called a {\it
path} from $x$ to $y$. The distance on the tree between two vertices $x$ and $y$ (denoted $d(x,y)$) is the length of the shortest edge-path
joining them.

Define
\[W_n := \{x\in V \quad \mid\quad d(x,o) = n \}\]
\[ \Lambda_{n}: = \bigcup_{j\le n}W_j;\quad  \Lambda_{[m,n]} = \bigcup_{j=m}^{n}W_j.\]

Recall a coordinate structure in $\G^k_+$:  every vertex $x$
(except for $x^0$) of $\G^k_+$ has coordinates $(i_1,\dots,i_n)$,
here $i_m\in\{1,\dots,k\}$, $1\leq m\leq n$ and for the vertex
$x^0$ we put $(0)$.  Namely, the symbol $(0)$ constitutes level 0,
and the sites $(i_1,\dots,i_n)$ form level $n$ (i.e. $d(x^0,x)=n$)
of the lattice.  Using this structure, vertices
$x^{(1)}_{W_n},x^{(2)}_{W_n},\cdots,x^{(|W_n|)}_{W_n}$ of $W_n$
can be represented as follows:
\begin{eqnarray}\label{xw}
&&x^{(1)}_{W_n}=(1,1,\cdots,1,1), \quad x^{(2)}_{W_n}=(1,1,\cdots,1,2), \ \ \cdots \quad x^{(k)}_{W_n}=(1,1,\cdots,1,k,),\\
&&x^{(k+1)}_{W_n}=(1,1,\cdots,2,1), \quad
x^{(2)}_{W_n}=(1,1,\cdots,2,2), \ \ \cdots \quad
x^{(2k)}_{W_n}=(1,1,\cdots,2,k),\nonumber
\end{eqnarray}
\[\vdots\]
\begin{eqnarray*}
&&x^{(|W_n|-k+1)}_{W_n}=(k,k,,\cdots,k,1), \
x^{(|W_n|-k+2)}_{W_n}=(k,k,\cdots,k,2),\ \ \cdots
x^{|W_n|}_{W_n}=(k,k,\cdots,k,k).
\end{eqnarray*}
In the above notations, we write
$$
 W_n = \{ (i_1, i_2, \cdots, i_n); \quad i_j = 1,2, \cdots, k \}
$$
So one can see that $|W_n|=k^n$.
The set of \textit{direct successors}  for a given vertex $x\in V$  is defined  by
\begin{equation}\label{S(x)def}
S(x) :  = \left\{y\in V \, \,  : \, \,  x\sim y \, \, \hbox{and} \, \, d(y,o) > d(x,o) \right\}.
\end{equation}
The vertex $x$ has exactly $k$ direct successors denoted $(x,i), i=1,2,\cdots, k$
$$
S(x) = \{(x,1),\,  (x,2), \, \cdots, \, (x,k)\}.
$$

To each vertex $x$, we associate a C$^*$--algebra of observable $\mathcal{A}_x$ with identity $\id_x$. For a given bounded region $V'\subset V$, we consider the algebra $\mathcal{A}_{V'} = \bigotimes_{x\in V'}\mathcal{A}_x$. We have the  the following natural embedding
$$
\mathcal{A}_{ \Lambda_{n }}\equiv  \mathcal{A}_{ \Lambda_{n }}\otimes\id_{ W_{n+1}}\subset \mathcal{A}_{ \Lambda_{n+1 }}.
$$
The algebra $\mathcal{A}_{ \Lambda_{n}}$ is then a subalgebra of $\mathcal{A}_{ \Lambda_{n+1}}$. It follows the local algebra
\begin{equation}\label{AVloc}
  \mathcal{A}_{V;\, loc} := \bigcup_{n\in\mathbb{N}}\mathcal{A}_{\Lambda_{n}}
\end{equation}
and the quasi-local algebra
$$
\mathcal{A}_V := \overline{\mathcal{A}_{V;\, loc}}^{C^*}
$$
The set of states on a C$^*$--algebra $\mathcal{A}$ will be denoted $\mathcal{S}(\mathcal{A})$.

There are  $k$ natural shifts on the Cayley tree  of order $k$: for each  $x = (i_1, i_2, \cdots, i_n) \in  \Lambda_n $ and $j\in\{1, \dots, k\}$
\begin{equation}\label{shifts}
  \alpha_j(x) = (j, x)= (j, i_1, i_2, \cdots, i_n)\in  \Lambda_{n+1}.
\end{equation}
 Let $g = (j_1, j_2, \cdots, j_N)\in V$ one defines
$$
\alpha_g(x) := \alpha_{j_1}\circ \alpha_{j_2}\circ\cdots\circ \alpha_{j_N}(x) = (j_1, j_2, \cdots, j_N, i_1, i_2,\cdots, i_n).
$$
The $\alpha_j$'s action on the algebra $\mathcal{A}_V$  is given as follows:
\begin{equation}\label{shift_algebra}
\alpha_j\left( \bigotimes_{ x\in  \Lambda_{\le n}}a_x \right):=
 \id^{(o)}\otimes\bigotimes_{x\in\Lambda_{\le n}} a_x^{(j,x)}.
\end{equation}
The shift $\alpha_j$ induces a  $*$-isomorphism from $\mathcal{A}_V$ into $\mathcal{A}_{V_{(o,j)}}$. Let  $\alpha_j^{-1}$ its inverse  isomorphism. For $g\in V$, the map $\alpha_g$ defines a  $*$-isomporphism from $\mathcal{A}_V$ into $\mathcal{A}_{V_g}$ and its inverse isomorphism will be denoted by $\alpha_g^{-1}$.

Consider a triplet $\mathcal C\subseteq\mathcal B\subseteq \mathcal A$ of C$^\ast$--algebras. A \textit{quasi-conditional expectation} \cite{ACe} is a completely positive identity preserving linear map $E :\mathcal A \to \mathcal B$ such that
$
E(ca) = cE(a)$, for all $a\in\mathcal A$, $c\in \mathcal C.
$
\begin{definition}\cite{ACe}
Let $\mathcal{B}\subseteq \mathcal{A}$ be two unitary C$^*$--algebra $\id$. A Markov transition expectation from $\mathcal{A}$ into $\mathcal{B}$ is a completely positive identity preserving map.
\end{definition}

 \begin{definition}\label{QMCdef}\cite{[AcFiMu07],[AcSouElG20]} A {\it (backward) quantum Markov chain}  on $\mathcal{A}_V$
 is a triplet $(\phi_o, (E_{\Lambda_{n}})_{n\ge 0},  (h_{n})_n)$
of initial state $\phi_o\in \mathcal{S}(\mathcal{A}_o)$, a sequence of quasi-conditional expectations $(E_{ \Lambda_{n }})_n$ w.r.t.
the triple $\mathcal{A}_{{\Lambda}_{n-1 }}\subseteq \mathcal{A}_{ \Lambda_{n }}\subseteq\mathcal{A}_{ \Lambda_{n+1}}$ and
a sequence $h_{n}\in\mathcal{A}_{W_n, +}$ of boundary conditions such that for each $a\in \mathcal{A}_V$  the limit
\begin{equation}\label{lim_Mc}
\varphi(a): = \lim_{n\to\infty} \phi_0\circ E_{ \Lambda_{0}}\circ
E_{ \Lambda_{1}} \circ \cdots \circ E_{ \Lambda_{n}}(h_{n+1}^{1/2}ah_{n+1}^{1/2})
\end{equation}
exists in the weak-*-topology and defines a state. In this case the state $\varphi$ defined by \eqref{lim_Mc}
is also called \textit{quantum Markov chain (QMC)}.
\end{definition}

A QMC $\varphi$ on $\mathcal{A}_{V}$ is said to be \textit{tree-homogeneous} if
 \begin{equation}\label{trans_inv_state}
   \varphi\circ \alpha_j = \varphi
 \end{equation}
for every $j \in \{1,2,\cdots, k\}$.

In the sequel, we restrict ourselves to the case of trivial boundary condition $h = \id$ and the associated tree-homogeneous quantum  Markov chain $\varphi$ is determined by the pair $\varphi\equiv (\phi_o, \mathcal{E})\equiv (\phi_o, \mathcal{E}, h= \id).$ \footnote{The existence of other boundary conditions leads to the problem of a phase transition within QMC scheme which was considered in \cite{MSH22,MBS161}.}
Here, $\mathcal{E}$ is a Markov transition expectation from $\mathcal{A}_{(o)}\otimes \mathcal{A}_{(1)}\otimes\cdots\otimes\mathcal{A}_{(k)}$ into $\mathcal{A}_{(o)}$. For each $u$ by $\mathcal{E}_u$ we denote the $\alpha_u$-shift of $\mathcal{E}$ given by
\begin{equation}\label{Eu=alphau(E)}
    \mathcal{E}_u = \alpha_{u}\circ\mathcal{E}\circ\alpha_{u}^{-1}
\end{equation}
Clearly, $\mathcal{E}_u$ is a transition expectation from $\mathcal{A}_{u}\otimes \mathcal{A}_{(u,1)}\otimes\cdots\otimes\mathcal{A}_{(u,k)}$ into $\mathcal{A}_u$.
For each $n\in\mathbb{N}$, we consider
$$
\mathcal{E}_{W_n}:= \bigotimes_{u\in W_n}\mathcal{E}_u
$$
One can see that $\mathcal{E}_{W_n}$ is a Markov transition expectation from $\mathcal{A}_{\Lambda_{[n,n+1]}}$ into $\mathcal{A}_{W_n}$. Following \cite{AF03,MS19}, we have the next result.

\begin{theorem}
Let $\varphi= (\phi_o, \mathcal{E})$ be a tree-homogeneous quantum Markov chain. There exists a unique conditional expectation $E_{o]}$ from $\mathcal{A}_V$ into $\mathcal{A}_{o}$ characterized by
\begin{equation}\label{Eo}
  E_{o]}(a) =   \mathcal{E}_{o}\left(a_{o}\otimes\mathcal{E}_{W_1}\left(a_{W_1}\cdots \otimes\mathcal{E}_{W_n}\left(a_{{W}_n}\otimes h_{n+1}\right)\right)\right)
\end{equation}
for all $a= a_o\otimes a_{W_1}\otimes\cdots\otimes a_{W_n}$. Moreover, one has
\begin{equation}\label{phi=phiooEo}
    \varphi(\cdot ) = \phi_o\circ E_{o]}( . )
\end{equation}
\end{theorem}


The forward Markov operator associated with $\mathcal{E}_u$ is given by:
\begin{equation}\label{Tu}
    T_u(a) = \mathcal{E}_u(a\otimes \id_{S(u)}), \quad    a\in \mathcal{A}_u
\end{equation}
While, there are  $k$ backward Markov operators corresponding to the successors $(u,\ell),\; j=1,\dots, k$  of $u$,
\begin{equation}\label{Pu}
    P_{u}^{(u,\ell)}(a) = \mathcal{E}_{u}(\id^{(u)}\otimes a\otimes \id_{S(u)\setminus\{(u,\ell)\}}),\quad \forall a\in \mathcal{A}_{(u,\ell)}
\end{equation}
For any ray $r=(u_n)_n$, one defines
\begin{equation}\label{Px0xm}
P_{u_n}^{u_{n+m}} = P_{u_n}^{u_{n+1}}\circ\cdots \circ P_{u_{n+m-1}}^{u_{n+m}}; \quad m,n\in\mathbb{N}
\end{equation}

The map $P_{u_n}^{u_m}$ defines a Markov operator from $\mathcal{A}_{u_m}$ into $\mathcal{A}_{u_n}$.

\section{Recurrence of quantum Markov chains on trees}

This section  is devoted to the notions of recurrence and weak recurrence for quantum Markov chains on trees.

Following \cite{accardi1,Sou22} a given projection $e\in Proj(\mathcal{A})$ and a ray $r=(u_n)_n\in Paths(o,\infty)$, a stopping time $\tau_{e;r} = (\tau_{u_n})_n$ on the algebra $\mathcal{A}_V$, is defined as follows:
\begin{eqnarray}
\tau_{e;o} &=& e^{(o)}\otimes \id_{V\setminus \{o\}} \nonumber\\
\tau_{e;u_1} &=& {e^\perp}^{(o)}\otimes e^{(u_1)}\otimes  \id_{V\setminus \{u_1]\}} \nonumber\\
\vdots&& \nonumber\\
\label{tauxn}\tau_{e;u_n} &=& {e^\perp}^{(o)}\otimes \cdots \otimes {e^{\perp}}^{(u_{n-1})}\otimes e^{(u_n)}\otimes\id_{V\setminus \{x_{n]}\}}
\end{eqnarray}
\begin{equation}\label{tauxinfty}
\tau_{e;u_n;\infty}:= {e^\perp}^{(o)}\otimes {e^{\perp}}^{(u_1)}\otimes \cdots\otimes {e^{\perp}}^{(u_{n-1})} \otimes {e^{\perp}}^{(u_n)}\otimes \id_{V\setminus \{u_n]},
\end{equation}
where for each $a\in\mathcal{A}$ one has $a^{(u)} = \alpha_u(a)$.
Put
$$
\tau_{e;r;\infty} =\lim_{n\to\infty}\tau_{e;u_n; \infty} =   \bigotimes_{n\in\mathbb{N}}{e^{\perp}}^{(u_{n})}
$$
\begin{definition} Let $\varphi = (\phi_o, \mathcal{E})$ be a tree-homogeneous quantum Markov chain. A projection $e\in Proj(\mathcal{A})$ is said to be
\begin{description}
\item[(i)] $\mathcal{E}$--completely accessible if
\begin{equation}
  E_{o]}(\tau_{e;r;\infty} ) : = \lim_{n\to\infty} E_{o]}(\tau_{e;x_n; \infty}) = 0
\end{equation}
for every ray $r=({x_n})_n$.
\item[(ii)] $\varphi$-completely accessible if $\varphi(\tau_{e; r;\infty})= 0$, for every ray $r=({x_n})_n$.
\item[(iii)] $\mathcal{E}$-recurrent if $0 <\mathrm{Tr}(\mathcal{E}(e\otimes\id))< \infty$ and one has
\begin{equation}
    \frac{1}{\mathrm{Tr}(\mathcal{E}(e\otimes\id))}\mathrm{Tr}\left(E_{o]}(\sum_{n\ge 0}e\otimes\tau_{e;x_n} \right) = 1
\end{equation}
for every ray $r=({x_n})_n$.
\item[(iv)] $\varphi$-recurrent if $\varphi(\alpha_o(e)) \ne 0$ and
\begin{equation}
    \frac{1}{\varphi(\alpha_o(e))}\varphi\left(\sum_{n}e\otimes\tau_{e; x_n}\right) = 1
\end{equation}
for every ray $r=({x_n})_n$.
\end{description}
\end{definition}
\begin{definition} Let $\varphi = (\phi_o, \mathcal{E})$ be a tree-homogeneous quantum Markov chain.
 Let $e,f\in Proj(\mathcal{A}), e,f\ne 0$. The projection $f$ is
 \begin{itemize}
\item[(i)] $\mathcal{E}$--accessible from $e$ (and we write $e\to^{\mathcal{E}} f$) if for any ray $r= (x_n)_n$ there exists $m\in\mathbb{N}$ such that
$$
E_{o]}\left(\alpha_0(e)\alpha_{x_m}(f)\right) \ne 0
 $$
 \item[(ii)] $\varphi$-accessible from $e$ (we denote it as $e\to^{\varphi} f$ if for any ray $r = (x_n)_n$ there exists  $m\in\mathbb{N}$ such that
 $$
 \varphi\left(\alpha_0(e)\alpha_{x_m}(f)\right) \ne 0
 $$
 \end{itemize}
 \end{definition}

\begin{lemma}\label{lem_sum}
In the above notations:
\begin{equation}
    \sum_{n\ge 0}\tau_{e;x_n} = \id_{\mathcal{A}_V} - \tau_{e;r;\infty}
\end{equation}
\end{lemma}
\proof see \cite{Sou22}
\begin{theorem}
Let $\varphi\equiv(\phi_o, \mathcal{E})$ be a tree-homogeneous quantum Markov chain on $\mathcal{A}_V$. Let $e\in Proj(\mathcal{A}_V)$ be a projection
\begin{description}
\item[(i)] $e$ is $\mathcal{E}$-recurrent if and only if for any ray $r= (x_n)_n$ one has
\begin{equation}\label{E_recurrence_eq}
  \mathcal{E}(e\otimes  E_{o]}(\tau_{e; r;\infty})) = 0
\end{equation}
\item[(ii)] $e$ is $\varphi$-recurrent if and only if for any ray $r= (x_n)_n$ one has
\begin{equation}\label{phi_recurr_eq}
    \varphi(e\otimes \tau_{e; r;\infty}) = 0
\end{equation}
\item[(iii)] $e$ is $\mathcal{E}$--accessible from $f$ if and only if for any ray $r= (x_n)_n$ there exists $m\in\mathbb{N}$ such that
\begin{equation}\label{E-recu-eq}
    \mathcal{E}(e\otimes P_{x_1}^{x_m}T_{x_m}f)\ne 0
\end{equation}
\item[(iv)] $e$ is $\varphi$--accessible from $f$ if and only if for any ray $r= (x_n)_n$ there exists $m\in\mathbb{N}$ such that
\begin{equation}
    \varphi(e\otimes P_{x_1}^{x_m}T_{x_m}f)\ne 0
\end{equation}
\end{description}
\end{theorem}
\begin{proof}
From Lemma \ref{lem_sum} one has
$$
\sum_{n\ge 0}e\otimes\tau_{x_n} = e\otimes\id - e\otimes\tau_{e;n;\infty}
$$
This leads to (i) and (ii).

One has
$$
\mathcal{E}_{W_n}\left(f^{(x_m)}\otimes \id\right)
= \mathcal{E}_{x_m}(f^{(x_m)}\otimes\id)  = T_{x_m}f
$$
and
$$
\mathcal{E}_{W_m}(\id_{W_{m-1}}\otimes (T_{x_n}f)^{(x_{m-1})}) = P_{x_{m-1}}^{x_m}T_{x_m}f
$$

\begin{eqnarray*}
E_{o]}\left(\alpha_0(e)\alpha_{x_m}(f)\right) &=&  \mathcal{E}_{W_0}(e\otimes\mathcal{E}_{W_1}(\id_{W_1}\otimes\cdots\otimes \mathcal{E}_{W_m}(\id_{W_{m-1}}\otimes\mathcal{E}_{W_m}(f^{(x_m)}\otimes\id_{W_m+1})\\
&=& \mathcal{E}_{W_0}(e\otimes\mathcal{E}_{W_1}(\id_{W_1}\otimes\cdots\otimes \mathcal{E}_{W_{m-2}}(\id_{W_{m-2}}\otimes (P_{x_{m-1}}^{x_m}T_{x_m}f)^{(x_{m-1})}\\
&&\vdots\\
&=& \mathcal{E}(e\otimes P_{x_1}^{x_m}T_{x_m}f)
\end{eqnarray*}
This proves (iii) and using (\ref{phi=phiooEo}) one gets (iv).
\end{proof}
\begin{corollary}
Let  $\varphi\equiv (\phi_o, \mathcal{E})$  be a tree-homogeneous quantum Markov chain. Any  $\mathcal{E}$-recurrence  projection  is $\varphi$-recurrent. Conversely, if the initial state $\phi_o$ is faithful then   Any  $\varphi$-recurrence  projection  is $\mathcal{E}$-recurrent.
\end{corollary}
\begin{proof}
Let $e\in Proj(\mathcal{A})$ be a projection. For each $\ell\in\{1,\dots, k\}$, one has
$$
E_{o]}(a_o\otimes \tau_{\ell}(a)) = \mathcal{E}(a_o\otimes E_{o]}(a)); \qquad  \forall a_o\in\mathcal{A }_o, \forall a\in\mathcal{A}_V
$$
Then
 \begin{eqnarray*}
 \varphi(e\otimes\tau_{e;r;\infty})) &=& \varphi(\alpha_o(e)\otimes\alpha_{(x_1)}(\tau_{e;r;\infty}))\\
 &\overset{(\ref{phi=phiooEo})}{=}& \phi_o\Big(E_{o]}\Big((\alpha_o(e)\otimes\alpha_{(x_1)}(\tau_{e;r;\infty})\Big)\Big)\\
 &=& \phi_o\Big( \mathcal{E}\Big(e\otimes{E}_{o]}\Big(\tau_{e;r;\infty}\Big)\Big)\Big)\\
 \end{eqnarray*}
 Therefore, if  $\mathcal{E}\Big(e\otimes{E}_{o]}\Big(\tau_{e;r;\infty}\Big)\Big) =0$ then $ \varphi(e\otimes\tau_{e;r;\infty}))=0$. This shows the first implication.

 If the initial state $\phi_o$ is faithful, since $\mathcal{E}\Big(e\otimes{E}_{o]}\Big(\tau_{e;r;\infty}\Big)\Big)\ge 0$ then from the above computation, we have
 $$
 \varphi(e\otimes\tau_{e;r;\infty})) =0 \Rightarrow \mathcal{E}\Big(e\otimes{E}_{o]}\Big(\tau_{e;r;\infty}\Big)\Big) = 0
 $$
 This shows the converse direction, and finishes the proof.
\end{proof}

 \section{Recurrence of QMC associated with OQRW}\label{QMC_tree}

  Let $\mathcal{H}$ and $\mathcal{K}$ be given two separable Hilbert spaces over the complex field $\mathbb{C}$. Let  $\{ |i\rangle\}_{i\in \Lambda}$ be an  ortho-normal basis of $\mathcal{K}$ indexed by a graph $\Lambda$ with almost-countable vertex set. The algebra of observable at a site $x\in V$ is considered to be  $\mathcal{A}_x = \mathcal{A}:=\mathcal{B}(\mathcal{H}\otimes\mathcal{K})$.

Let $\mathcal{M}$ be a OQRW given by \eqref{MM1}.
In the language of OQRW \cite{attal} the Hilbert space $\mathcal{H}$ describes the internal degree of freedom of the quantum walker, while $\mathcal{K}$ describes the state space of the dynamics where the walk is dome through the oriented graph $\Lambda$.
The transition of the walker from a site $j$ to site $i$ is described by a bounded operator  $B_{j}^{i}\in \mathcal{B}(\mathcal{H})$  such that
 \begin{equation}\label{sumBB=1}
 \sum_{i\in\Lambda}B_j^{i*}B_{j}^{i} = \id_{\mathcal{B}(\mathcal{H})}.
 \end{equation}
The initial density matrix of the dynamics is  $\rho\in \mathcal{B}(\mathcal{H}\otimes\mathcal{K})$, of the form
 $$
 \rho = \sum_{i\in\Lambda}\rho_i\otimes |i\rangle\langle  i|; \quad \rho_i\in\mathcal{B}(\mathcal{H})^{+}.
 $$
In what follows, for the sake of simplicity of calculations,  we assume that $\r_i\neq 0$ for all $i\in\Lambda$ (see \cite[Remark 4.5]{DM19} for other kind of initial states).

Let
 \begin{equation}\label{Mij}
 M_j^i = B_j^i\otimes|i\rangle\langle j| \in\mathcal{B}(\mathcal{H}\otimes \mathcal{K}).
 \end{equation}
 and
 \begin{equation}\label{Aij}
 A_{j}^{i} := \frac{1}{{\tr(\rho_j)}^{1/2}}\rho_j^{1/2}\otimes |i\rangle\langle j|, \quad i,j\in \Lambda.
 \end{equation}
 For each $u\in V$, we define
\begin{equation}\label{Kji}
 {K_{j}^{i}}^{(u,S(u))} := {M_{j}^{i*}}^{(u)}\otimes\bigotimes_{v\in S(u)}{A_{j}^{i}}^{(v)} \in\mathcal{A}_{\{u\}\cup S(u)}.
\end{equation}
The interaction of a vertex $u\in V$ with its set of direct successors it describled by
\begin{equation*}
    K^{(u,S(u))} = \sum_{i,j}{K_{j}^{i}}^{(u,S(u))} \in\mathcal{A}_{\{u\}\cup S(u)}
\end{equation*}
Put
\begin{equation}\label{Eu_def}
\mathcal{E}_u(a)
:= \tr_{u]}( K^{(u,S(u))}a K^{(u,S(u))*}); \quad a\in\mathcal{A}_{\{u\}\cup S(u)}.
\end{equation}
For each $j,j'\in\Lambda$ we set
\begin{equation}\label{phijj'}
\varphi_{jj'}(b):= \frac{1}{\mathrm{Tr}(\rho_j)^{1/2}\tr(\rho_{j'})^{1/2}}\tr\left(\rho_j^{1/2} \rho_{j'}^{1/2} \otimes |j'\rangle\langle j|\, b\right); \quad \forall a\in\mathcal{A}
\end{equation}
One can see that $\varphi_{jj'}$ is a linear functional on $\mathcal{A}$. If $j=j'$,  we denote it simply denote $\varphi_j$ instead of $\varphi_{jj}$ one has
\begin{equation}\label{phij}
    \varphi_{j}(a)  = \frac{1}{\tr(\rho_j)}\tr\Big(\rho_j\otimes|j\rangle\langle j| a\Big)
\end{equation}
The functional $\varphi_j$ is then, a state on $\mathcal{A}$.
\begin{theorem}\label{thmEuTuPu}
In the above notations, the map $\mathcal{E}_{u}$ defines a Markov transition expectation from $\mathcal{A}_{\{u\}\cup S(u)}$ into $\mathcal{A}_u$ and
\begin{equation}\label{Eu}
    \mathcal{E}_{u}(a_{u}\otimes a_{(u,1)}\otimes\cdots\otimes a_{(u,k)}) =   \sum_{(i,j ,j')\in \Lambda^3} M_j^{i*} a_{(u)}M_{j'}^{i}\prod_{\ell=1}^{k} \varphi_{jj'}(a_{_{(u,\ell)}})
\end{equation}
Moreover, the backward Markov operators associated with $\mathcal{E}_u$ are given by
\begin{equation}
    P_{u}^{(u,\ell)}(a_{(u,\ell)}) = \sum_{j}\Big(\id_{\mathcal{B}(\mathcal{H})}\otimes|j\rangle\langle j|\Big)\varphi_{j}(a_{(u,\ell)})
\end{equation}
The forward Markov operator associated with $\mathcal{E}_u$ is given by
\begin{equation}
 T_{u}(a_u) = \sum_{ij}M_{j}^{i,*}a_u M_{j}^{i}
\end{equation}
where $a_u\in\mathcal{A}$ and $a_{(u,\ell)}\in\mathcal{A}_{(u,\ell)}$ for each $\ell \in\{1,\cdots, k\}.$
\end{theorem}

\begin{proof}  The map $\mathcal{E}_u$  (\ref{Eu_def}), is clearly completely positive. \\
Let $a =a_{u}\otimes a_{(u,1)}\otimes\cdots\otimes a_{(u,k)}$. Taking into account (\ref{Kji}) and  (\ref{Mij}) one gets
\begin{eqnarray*}
\mathcal{E}_{u}(a)
&=&  \tr_{u]}\left( \left(\sum_{(i,j) \in\Lambda^2}{K_{j}^{i}}\right) a \left(\sum_{  (i,j)\in\Lambda^2}{ K_{j}^{i}}\right)^*\right)\\
&=& \sum_{(i,j), (i',j')\in\Lambda^2} \tr_{u]}\left( {K_{j}^{i}}^{(u,S(u))} a_{u}\otimes a_{(u,1)}\cdots\otimes a_{(u,k)}{ K_{j'}^{i'}}^{(u,S(u))\; *}\right)\\
 &= &
 \sum_{(i,j),(i',j')\in \Lambda^2}\tr_{u]}\left( {M_{j}^{i}}^{(u)\, *} a_u{ M_{j'}^{i'}}^{(u)}\otimes\bigotimes_{\ell=1}^{k}\left(A_{j}^{i}a_{(u,\ell)}A_{j'}^{i'*}\right)^{(u,\ell)}\right)\\
&=&\sum_{(i,j),(i',j')\in \Lambda^2} M_j^{i*}a_0^{(u, 0)}M_{j'}^{i'} \prod_{\ell=1}^{k} \tr( A_{j}^{i}a_{\ell}^{(u,\ell)}A_{j'}^{i'*})\\
\end{eqnarray*}
  For each $\ell\in\{1,\dots, k\}$,  one has
\begin{eqnarray*}
\tr( A_{j}^{i}a_{(u,\ell)}A_{j'}^{i'*}) &\overset{(\ref{Aij})}{=}& \tr_Big( A_{j'}^{i'*}A_{j}^{i}a_{(u,\ell)}\Big)\\
&=& \frac{1}{\tr(\rho_j)^{1/2}\tr(\rho_{j'})^{1/2}}\tr\Big(\rho_{j'}^{1/2}\rho_{j}^{1/2}\otimes |j'\rangle\langle   j|a_{(u,\ell)}\Big)\delta_{i,i'}\\
&\overset{(\ref{phijj'})}{=}& \varphi_{jj'}(a_{(u,\ell)}) \delta_{i,i'}
\end{eqnarray*}
where $\delta_{i,i'}$ denotes the Kronecker symbol. This leads to (\ref{Eu}). One has
$$
\mathcal{E}_{u}(\id_{(u, S(u))}) = \sum_{i,j,j'}M_{j}^{i\;*}M_{j'}^{i}\prod_{\ell=1}^{k}\varphi_{jj'}(\id_{(u,\ell)})\\
\overset{(\ref{phijj'})}{=}\sum_{i,j}M_{j}^{i\, *}M_{j}^{i} = \id_{u}
$$
Then $\mathcal{E}_u$ is a Markov transition expectation.

From (\ref{Pu}) one has
\begin{eqnarray*}
P_{u}^{(u,\ell)}(a_{(u,\ell)} &=& \sum_{i,j,j'} M_{j}^{i\;*}M_{j'}^{i}\varphi_{jj'}(a_{(u,\ell)})\prod_{\underset{\ell'\ne\ell}{\ell'=1}}^{k}\varphi_{jj'}(\id_{(u,\ell')})\\
&=& \sum_{i,j}M_{j}^{i\;*}M_{j}^{i}\varphi_{j}(a_{(u,\ell)})\\
&=& \sum_{j}\Big(\sum_{i}B_{j}^{i\;*}B_{j}^{i}\Big)\varphi_{j}(a_{(u,\ell)})\\
&\overset{(\ref{sumBB=1})}{=}& \sum_{j}\Big(\id_{\mathcal{H}}\otimes|j\rangle\langle j|\Big)\varphi_{j}(a_{(u,\ell)})
\end{eqnarray*}
The forward Markov operator (\ref{Tu}) associated with $\mathcal{E}_u$ satisfies
$$
T_u(a_u)= \sum_{i,j,j'} M_{j}^{i\, *}a_uM_{j'}^{i}\prod_{\ell=1}^{k}\varphi_{jj'}(\id_{(u,\ell)})  = \sum_{i,j} M_{j}^{i\, *}a_uM_{j}^{i}
$$
This finishes the proof.
\end{proof}
Now we are ready to Build the conditional expectation $E_{o]}$  in the case of open quantum random walks using the transition expectations of the form (\ref{Eu_def}) and the quantum Markov chain $\varphi\equiv (\phi_o, \mathcal{E})$, where
\begin{equation}\label{calE}
\mathcal{E}(a):= \mathcal{E}_o(a) = \sum_{i,j}M_{j} ^{i\, *}a_oM_{j}^{i}\prod_{\ell=1}^{k}\varphi_j(a_{(o,\ell)})
\end{equation}
for each $ a= a_{o}\otimes a_{(o,1)}\otimes\cdots\otimes a_{(o,k)}$.\\
 It is clear that for each $u\in V$ the transition expectation $\mathcal{E}_u$ is a copy of $\mathcal{E}$ in the sense of (\ref{Eu=alphau(E)}).
\begin{theorem}\label{thmEoOQRW}
 In the above notations, the conditional expectation associated $E_{o]}$ associated with $\mathcal{E}$ through (\ref{Eo}) has the following expression
 \begin{equation}\label{Eo_expr}
     E_{o]}(a) = \sum_{j}\mathcal{M}_{j}(a_o)\prod_{u\in\Lambda_{[1,n]}}\psi_{j}(a_u)
 \end{equation}
 where
 \begin{equation}\label{psi_j}
  \psi_{j}(b) = \frac{1}{\tr(\rho_j)}\sum_{i\in\Lambda}\tr\left(B_{j}^{i}\rho_j{B_{j}^{i}}^{*}\otimes |i\rangle\langle i|b\right), \quad  \forall b\in\mathcal{A}.
\end{equation}
and $a= \bigotimes_{u\in\Lambda_n}a_u\in\mathcal{A}_{\Lambda_n}$. Moreover, for any initial state $\phi_o = \tr(\omega_o \cdot  )$ the tree-homogeneous quantum Markov chain $\varphi\equiv (\phi_o, \mathcal{E})$ is given by
\begin{equation}\label{varphi_id}
    \varphi(a) = \sum_{j} \tr\left(\omega_o)\mathcal{M}_{j}( a_o) \right)\prod_{u\in \Lambda_{[1,n]}}\psi_{j}(a_{u})
\end{equation}
where
\begin{equation}\label{Mj}
   \mathcal{M}_{j}( \cdot ) = \sum_{i\in\Lambda} M_{j}^{i*} \, \cdot\,M_j^{i}
\end{equation}
\end{theorem}

\begin{remark}
We notice that in our previous work \cite{MSH22} the expression (\ref{varphi_id})  defines the QMC associated with the disordered phase of the system  that deals with phase transitions for QMC on trees associated with OQRW.
 \end{remark}
\begin{theorem}\label{thmErecurrent} In the notations of Theorem \ref{thmEoOQRW}, if $e$ is a projection in $\mathcal{A}$ such that
\begin{equation}
p:= \sup_{j\in\Lambda}\psi_j(e^{\perp}) < 1
\end{equation}\label{suppsi}
then  $e$ is $\mathcal{E}$-recurrent.
\end{theorem}
\begin{proof}
Let $r= (x_n)_n$ be a ray one the semi-infinite Cayley tree.
One has
\begin{eqnarray*}
E_{o]}(\tau_{e;x_n;\infty}) &\overset{(\ref{Eo_expr})}{=}& \sum_{j\in\Lambda}M_{j}^{i\,*} \alpha_o(e^{\perp}) M_{j}^{i}\prod_{m=1}^{n}\psi_{j}(\alpha_{x_m}(e^\perp))\\
&\overset{(\ref{psi_j})}{=}& \sum_{j\in\Lambda}M_{j}^{i\,*} e^{\perp} M_{j}^{i}\left(\psi_{j}( e^\perp)\right)^n\\
&\leq& \sum_{j}M_{j}^{i\, *}M_{j}^{i}\,  p^n\\
&=&   p^n
\end{eqnarray*}
From (\ref{suppsi})  one gets
$$
0 \le E_{o]}(\tau_{e;r;\infty}) = \lim_{n\to\infty}E_{o]}(\tau_{e;x_n;\infty})  = 0
$$
Thus $E_{o]}(\tau_{e;r;\infty}) =0$ and by (\ref{E_recurrence_eq}) the projection $e$ is $\mathcal{E}$-recurrent.
\end{proof}
   \section{Examples}\label{Sect_exp}
In this section, we are going to illustrate the obtained results on recurrence for quantum Markov chains associated with OQRW.

Let  $\mathcal{H}= \mathcal{K} = \mathbb{C}^2$.  The algebra of observable at a site $u$ is then  $\mathcal{A}_u = \mathcal{B}(\mathcal{H}\otimes\mathcal{H}) \equiv M_4(\mathbb{C})$. Let $\Lambda = \{1,2\}$. The interactions are given by
  \begin{equation}\label{Bjimodel}
     B_1^1 = \left(
               \begin{array}{cc}
                 a & 0 \\
                 0 & b \\
               \end{array}
             \right), \quad  B_2^1 = \left(
               \begin{array}{cc}
                 0 & 1 \\
                 0 & 0 \\
               \end{array}
             \right), \quad  B_1^2 = \left(
               \begin{array}{cc}
                 c & 0 \\
                 0 & d \\
               \end{array}
             \right),\quad  B_2^2 = \left(
               \begin{array}{cc}
                 1 & 0 \\
                 0 & 0 \\
               \end{array}
             \right)
   \end{equation}
   where  $|a|^2 + |c|^2 = |b|^2 + |d|^2  =1, ac\ne 0$.
   Put \begin{equation}\label{pq}
p=\left(
      \begin{array}{cc}
        1 & 0 \\
       0 & 0 \\
      \end{array}
    \right), \qquad q=\left(
      \begin{array}{cc}
        0 & 0 \\
       0 & 1 \\
      \end{array}
    \right).
\end{equation}
and
$$|1\rangle = \left[\begin{array}{cc} 1 \\ 0  \end{array}\right] , |2\rangle  = \left[\begin{array}{ll}0 \\ 1  \end{array}\right] $$
 Notice that  $(|1\rangle,|2\rangle)$ is an ortho-normal basis of $\mathcal{K}\equiv \mathbb{C}^2$. In  the sequel elements of $\mathcal{B}(\mathcal{H})$ will be denoted by means of $2\times2$ complex matrices, while elements of $\mathcal{B}(\mathcal{K})$ will be written using Dirac notation $|i\rangle\langle j|.$

Recall that (c.f. \cite{}) any rank-1 projection in $\mathbb{M}_2(\mathbb{C})$ has the form
\begin{equation}\label{ez}
 p(\varepsilon,z) = \left(\begin{array}{cc}  \varepsilon & z\sqrt{\varepsilon(1-\varepsilon)}\\
 \\
  \overline{z} \sqrt{\varepsilon(1-\varepsilon)}   & 1-\varepsilon \\
 \end{array}\right)
 \end{equation}
 where $\varepsilon\in[0,1], z\in \mathbb{C}$ with $|z| =1$.
Then we consider the projection on $\mathcal{A}$ having the following form
$$
e(\varepsilon, z, \xi) = p(\varepsilon, z)\otimes |\xi\rangle\langle\xi|
$$
where $$|\xi\rangle:= \sum_{i\in\Lambda}\xi_i|i\rangle\in\mathcal{K}$$  being a unit vector. i.e. $\sum_{i\in\Lambda}|\xi_i|^2 = 1$.
\begin{example}[$\mathcal{E}$-recurrence]
Using (\ref{psi_j}) one compute
\begin{align*}
\psi_j(e(\varepsilon, z, \xi)) =& \frac{1}{\tr(\rho_j)}\sum_{i\in\Lambda}\tr\left( B_{j}^{i} \rho_j B_{j}^{i\; *} p(\varepsilon,z)\right)|\xi_i|^2
\end{align*}

Then, for $\rho_j=\left(
               \begin{array}{cc}
                 1 & 0 \\
                 0 & 0 \\
               \end{array}
             \right)$, one gets
\begin{equation*}
 \tr\left( B_{1}^{1} \rho_j B_{1}^{1\; *} p(\varepsilon, z)\right) =  \varepsilon |a|^2 ,
\end{equation*}
\begin{equation*}
 \tr\left( B_{2}^{1} \rho_2 B_{2}^{1\; *} p(\varepsilon, z)\right) =   0,
\end{equation*}
\begin{equation*}
 \tr\left( B_{1}^{2} \rho_j B_{1}^{2\; *} p(\varepsilon, z)\right) =   \varepsilon |c|^2 ,
\end{equation*}
\begin{equation*}
 \tr\left( B_{2}^{2} \rho_j B_{2}^{2\; *} p(\varepsilon, z)\right) =  \varepsilon .
\end{equation*}

Hence,
\begin{align}\label{psi1(e)}
\psi_1(e(\varepsilon, z, \xi)) =& \sum_{i\in\Lambda}\tr\left( B_{1}^{i} \rho_j B_{1}^{i\; *} p(\varepsilon, z)\right)|\xi_i|^2
=&  \varepsilon |a|^2 |\xi_1|^2.
\end{align}
and
\begin{align}\label{psi2(e)}
\psi_2(e(\varepsilon, z, \xi)) =& \sum_{i\in\Lambda}\tr\left( B_{2}^{i} \rho_j B_{2}^{i\; *} p(\varepsilon, z)\right)|\xi_i|^2
=&  \varepsilon |c|^2 |\xi_1|^2 +\varepsilon |\xi_2|^2 .
\end{align}
Thus, Theorem \ref{thmErecurrent} implies that
  $e(\varepsilon, z, \xi)^{\perp}$ is $\mathcal{E}$-recurrent whenever $\varepsilon <1$.
 If $\varepsilon=|a|=|\xi_1|=1$, the projection $e(\varepsilon, z, \xi)$ becomes
 \begin{equation*}
    e(1, z, \xi) =   \left(
               \begin{array}{cc}
                 1 & 0 \\
                 0 & 0 \\
               \end{array}
             \right)\otimes |1\rangle\langle1|.
 \end{equation*}
 Put $$e := e(1, z, \xi)^{\perp} = \id_{M_2}\otimes |2\rangle\langle2|+\left(
               \begin{array}{cc}
                 0 & 0 \\
                 0 & 1 \\
               \end{array}
             \right)\otimes |1\rangle\langle1|$$
 From (\ref{psi1(e)}) and (\ref{psi2(e)}) one has $\psi_1(e^\perp) = 1$ and $\psi_2(e^{\perp}) = 0.$
 Then, from (\ref{Eo_expr}) one gets
 \begin{align*}
E_{o]}(\tau_{e ;x_n;\infty}) =\mathcal{M}_1(e^{\perp})=& \sum_{i=1}^2M_{1}^{i\,*} e^{\perp} M_{1}^{i}
\\=&{B_1^1}^*\left(
               \begin{array}{cc}
                 1 & 0 \\
                 0 & 0 \\
               \end{array}
             \right)B_1^1\otimes |1\rangle\langle1|
             \\=&e^{\perp}
\end{align*}

 Therefore,
 \begin{align*}
     \mathcal{E}(e\otimes  E_{o]}(\tau_{e; r;\infty})) =&\mathcal{E}(e\otimes  e^{\perp})
     \\=&\mathcal{M}_1(e)\psi_1(e^{\perp})
      \\=&\mathcal{M}_1(e)
      \\=& \left(
               \begin{array}{cc}
                 0 & 0 \\
                 0 & |b|^2 \\
               \end{array}
             \right)\otimes |1\rangle\langle1|+\left(
               \begin{array}{cc}
                 0 & 0 \\
                 0 & |d|^2 \\
               \end{array}
             \right)\otimes |2\rangle\langle2|\ne 0
 \end{align*}
 Thus, from  \eqref{E-recu-eq} the projection $e$ is not $\mathcal{E}$-recurrent. This means that the inequality (\ref{suppsi}) is optimal.
  \end{example}
  \begin{example}[$\mathcal{E}$-accessibility]
Recall that for $\ell =1,2$, the backward Markov operator is given by
$$
 P_{u}^{(u,\ell)}(a_{(u,\ell)}) = \sum_{j=1}^{2}\Big(\id_{\mathcal{B}(\mathcal{H})}\otimes|j\rangle\langle j|\Big)\varphi_{j}(a_{(u,\ell)}).
$$
Recall also that the forward Markov operator is given by
$$
T_{u}(a_u) = \sum_{ij}{}M_{j}^{i,*}a_uM_{j}^{i}.
$$
Then,
\begin{equation*}
 \mathcal{E}\left(e\otimes P_{x_0}^{x_m}T_{x_m}f\right) =\sum_{j}\psi_{j}(f)\mathcal{E}\left(e\otimes\id\otimes|j\rangle\langle j|\right)
\end{equation*}
\begin{itemize}
    \item Take $e\in Proj(\mathcal{A})$ and $f=e(\varepsilon, z, \xi)$, then using \eqref{psi1(e)} and \eqref{psi2(e)}, we obtain
 \begin{align*}
 \mathcal{E}\left(e\otimes P_{x_0}^{x_m}T_{x_m}e(\varepsilon, z, \xi)\right) &=\sum_{j}\psi_{j}(e(\varepsilon, z, \xi))\mathcal{E}\left(e\otimes\id\otimes|j\rangle\langle j|\right)
 \\&=\varepsilon |a|^2 |\xi_1|^2 \mathcal{E}\left(e\otimes\id\otimes|1\rangle\langle 1|\right)  +\varepsilon( |c|^2 |\xi_1|^2 + |\xi_2|^2)\mathcal{E}\left(e\otimes\id\otimes|2\rangle\langle 2|\right)
 \\&=\varepsilon \left[|a|^2 |\xi_1|^2 \mathcal{M}_1(e) +( |c|^2 |\xi_1|^2 + |\xi_2|^2)\mathcal{M}_2(e)\right]
\end{align*}
for any projection $e$. In particular, one easily can see that there is no projection $e$ which is $\mathcal{E}$-accessible from
\begin{equation*}
   e(0, z, \xi)= \left(
               \begin{array}{cc}
                 0 & 0 \\
                 0 & 1 \\
               \end{array}
             \right)\otimes |\xi\rangle\langle\xi|.
\end{equation*}
    \item Now, take
    \begin{equation*}
    f=\sigma^{x_{W_n}(1)} =  \id_{M_2}\otimes |1\rangle\langle1|
\end{equation*}
 where $x_{W_n}(1)$ is defined by (\ref{xw}). Then, we have
 \begin{equation*}
     \psi_{1}(\sigma^{x_{W_n}(1)})=\tr(B_1^1p{B_1^1}^*)=|a|^2\quad{\rm and} \ \psi_{2}(\sigma^{x_{W_n}(1)})=\tr(B_2^1p{B_2^1}^*)=0.
 \end{equation*}
 Hence,
 \begin{align*}
 \mathcal{E}\left(e\otimes P_{x_1}^{x_m}T_{x_m}\sigma^{x_{W_n}(1)}\right) =&|a|^2\mathcal{E}\left(e\otimes (\id\otimes|1\rangle\langle 1|)^{(x_1)}\right)
\\=& |a|^2 \sum_{i}{M_{1}^{i}}^{*}e M_{1}^{i}
\\=&|a|^2\mathcal{M}_1(e).
\end{align*}
In particular, if $|a|>0$, we deduce that
\begin{equation*}
     e_1=\left(
               \begin{array}{cc}
                 1 & 0 \\
                 0 & 0 \\
               \end{array}
             \right)\otimes |1\rangle\langle1|
 \end{equation*}
 is $\mathcal{E}$-accessible from $\sigma^{x_{W_n}(1)}$, since $\mathcal{M}_1(e_1)=e_1$.
\end{itemize}
 \end{example}

 \begin{example}[$\varphi$-accessibility]

 We notice that,
 \begin{equation*}
 \varphi\left(e\otimes P_{x_0}^{x_m}T_{x_m}f\right) =\sum_{j}\psi_{j}(f)\varphi\left(e\otimes\id\otimes|j\rangle\langle j|\right)
 \end{equation*}
 where
 \begin{align*}
     \varphi\left(e\otimes\id\otimes|j\rangle\langle j|\right) &= \sum_{k} \tr\left(\omega_o\mathcal{M}_{k}(e) \right)\psi_{k}(\id\otimes|j\rangle\langle j|)
     \\&= \sum_{k} \tr\left(\omega_o\mathcal{M}_{k}(e) \right)\psi_{k}(\id\otimes|j\rangle\langle j|)
     \\&= \sum_{k} \tr\left(\omega_o\mathcal{M}_{k}(e) \right)\frac{Tr(B_k^j\rho_k{B_k^j}^*)}{\tr(\rho_k)}.
 \end{align*}
 Hence, for $|a|>0$ and
\begin{equation*}
     \omega_0=\left(
               \begin{array}{c|c}
                 (0) & (0)  \\
                 \hline
                  (0) & (*) \\
               \end{array}
             \right)
 \end{equation*}
 we deduce that
\begin{equation*}
     e_1=\left(
               \begin{array}{cc}
                 1 & 0 \\
                 0 & 0 \\
               \end{array}
             \right)\otimes |1\rangle\langle1|
 \end{equation*}
 is not $\varphi$-accessible from $\sigma^{x_{W_n}(1)}$, since
 $$\mathrm{Tr}\left(\omega_o\mathcal{M}_{1}(e_1) \right)=\mathrm{Tr}\left(\omega_oe_1 \right)=0\quad \hbox{and} \quad \mathcal{M}_2(e_1)=0$$
 \end{example}

\section*{Declaration of Competing Interest}

The authors confirm that there are no known conflicts of interest associated with this publication and there has been no significant financial support for this work that could have influenced its outcome.

\section*{Data availability}
The paper does not use any data.

\section*{Acknowledgments}
The authors gratefully acknowledge Qassim University, represented by the Deanship of Scientific Research, on the
financial support for this research under the number (10173-cba-2020-1-3-I)
during the academic year 1442 AH / 2020 AD.


\begin{thebibliography}{99}

%
%

\bibitem{Ac} L. Accardi, On noncommutative Markov property,
\textit{Funct. Anal. Appl.} {\bf 8} (1975), 1--8.


\bibitem{[AcFr80]}
L. Accardi, A. Frigerio, Markovian cocycles, \emph{Proc. Royal Irish
Acad.} {\bf 83A} (1983) 251-263.


\bibitem{ACe} L. Accardi, C. Cecchini, Conditional expectations in von Neumann algebras and a Theorem of Takesaki,
\textit{J. Funct. Anal.} {\bf 45}, 245--273 (1982).

\bibitem{AF03} Accardi L., Fidaleo F., Non homogeneous quantum Markov states and
quantum Markov fields, \textit{J. Funct. Anal.} {\bf 200} (2003),
324-–347.

\bibitem{[AcFiMu07]}
L. Accardi, F. Fidaleo, F. Mukhamedov, Markov states and chains on
the CAR algebra, \textit{Inf. Dim. Analysis, Quantum Probab.
Related Topics} {\bf 10} (2007), 165--183.

\bibitem{accardi1} L. Accardi, D. Koroliuk, Stopping times for quantum Markov chains, \textit{J. Theor. Probab.} {\bf 5}(1992), 521-535.

\bibitem{accardi2} L. Accardi, D. Koroliuk, Quantum Markov chains: The recurrence problem. In book: \textit{Quantum Prob. and Related Topics VII}, 63--73 (1991).


\bibitem{[AcSouElG20]} L. Accardi, A. Souissi, E. Soueidy, Quantum Markov chains: A unification approach,
 \textit{Inf. Dim. Analysis, Quantum Probab. Related Topics } {\bf 23}(2020), 2050016.

%
%
%
%


\bibitem{AOM} L. Accardi, H. Ohno, F. Mukhamedov, Quantum Markov fields on
graphs, \textit{Inf. Dim. Analysis, Quantum Probab. Related Topics}
{\bf 13}(2010), 165--189.

%

\bibitem{AW87}
L. Accardi, G.S. Watson, Quantum random walks, in book: L. Accardi, W. von Waldenfels (eds)  \textit{Quantum Probability and Applications IV, Proc. of the year of Quantum Probability, Univ. of Rome Tor
Vergata, Italy, 1987},  LNM, {\bf 1396}(1987), 73--88.


\bibitem{attal} S. Attal, F. Petruccione, C. Sabot, I. Sinayskiy. Open Quantum Random Walks. \textit{J. Stat. Phys.} {\bf 147}(2012), 832-852.



\bibitem{BBP} I. Bardet, D. Bernard, Y. Pautrat, Passage times, exit times and Dirichlet problems for open quantum walks,
\textit{J. Stat. Phys.} {\bf 167}(2017), 173-204.




\bibitem{burgarth2} D. Burgarth, V. Giovannetti, The generalized Lyapunov theorem and its application to quantum channels.\textit{ New J.  Phys.} 9 (2007) 150.

\bibitem{carbone} R. Carbone, Y. Pautrat. Homogeneous open quantum random walks on a lattice. \textit{J. Stat.
Phys.} {\bf 160}(2015), 1125-1152.

\bibitem{carbone2} R. Carbone, Y. Pautrat. Open quantum random walks: reducibility, period, ergodic properties. \textit{Ann. Henri
Poincar\'e} {\bf 17}(2016), 99-135.

\bibitem{CGL17} S.L. Carvalho, L.F. Guidi, C.F. Lardizabal, Site recurrence of open and unitary quantum walks on the line, \textit{Quantum Infor. Proc.}{\bf 16}(2017), Article 17.


\bibitem{CV} J.I. Cirac, F. Verstraete, Renormalization and tensor product states in spin chains and
lattices, \textit{J. Phys. A. Math. Theor.} {\bf 42} (2009),
504004.

%

%
%


%

\bibitem{DK19} A. Dhahri. C.K. Ko, H.J. Yoo, Quantum Markov chains associated with open quantum random walks, \textit{J.
Stat. Phys.} {\bf 176}(2019), 1272–1295


\bibitem{DM19} A. Dhahri, F. Mukhamedov, Open quantum random walks, quantum Markov chains and recurrence. \textit{Rev. Math. Phys.} {\bf 31}(2019), 1950020.




\bibitem{Fing} Y. Feng, N. Yu and M. Ying, Model checking quantum Markov chains, \textit{J. Computer Sys. Sci.} {\bf 79}, 1181–-1198 (2013).





\bibitem{GLV20}  F.A. Grünbaum, C.F. Lardizabal, L.Vel\'{a}zquez, Quantum Markov chains: recurrence, Schur functions and splitting rules,
\textit{Ann. Henri Poincare}, {\bf 21}(2020), 189–239.


\bibitem{JL21} T. S. Jacq, C. F. Lardizabal, Homogeneous open quantum walks on the line:
criteria for site recurrence and absorption, \textit{Quantum Inf. Comput.}, {\bf 21}(2021), 37–58.

\bibitem{Ke} J. Kempe, Quantum random walks—an introductory overview, \textit{ Contemporary Physics}, {\bf 44}(2003), 307–327.

%


\bibitem{KKS19} C. K. Ko, N. Konno, E. Segawa, H. J. Yoo. Central limit theorems for open
quantum random walks on the crystal lattices, \textit{J. Stat. Phys.} {\bf 176}(2019), 710–735.


\bibitem{konno} N. Konno, H. J. Yoo. Limit theorems for open quantum random walks. \textit{J. Stat. Phys.} {\bf 150} (2013), 299-319.

\bibitem{Kum} B. K\"{u}mmerer, Quantum Markov processes and applications in physics. In book: Quantum independent increment processes. II,  259--330, \textit{ Lecture Notes in Math.}, 1866, Springer, Berlin, 2006.

\bibitem{cfrr} C. F. Lardizabal, R. R. Souza. On a class of quantum channels, open random walks and recurrence. \textit{J. Stat. Phys.} {\bf 159}(2015), 772-796.




\bibitem{petulante} C. Liu, N. Petulante. On Limiting distributions of quantum Markov chains. \textit{Int. J. Math. and Math. Sciences.} {\bf 2011}(2011), ID 740816.




\bibitem{MSK13}  A. Marais, I. Sinayskiy, A. Kay, F. Petruccione,  A. Ekert,  Decoherence-assisted
transport in quantum networks, \textit{New J. Phys.}, {\bf 15}(2013), 013038.




\bibitem{MBS161} F. Mukhamedov, A. Barhoumi,  A. Souissi, Phase transitions for quantum Markov chains associated with Ising type models on a Cayley tree, \textit{J. Stat. Phys.} {\bf 163}, 544--567 (2016).



\bibitem{MBSG20} F. Mukhamedov, A. Barhoumi, A. Souissi, S. El Gheteb,
A quantum Markov chain approach to phase
transitions for quantum Ising model
with competing XY-interactions on a Cayley tree, \textit{J. Math. Phys.} {\bf 61}, 093505 (2020).


%

%
%


\bibitem{MS19} F. Mukhamedov, A. Souissi, Quantum Markov States on Cayley trees, \textit{J. Math. Anal. Appl.} {\bf 473}(2019), 313--333.


\bibitem{MSH22} F. Mukhamedov, A. Souissi, T. Hamdi,  Open Quantum Random Walks and Quantum Markov chains on Trees I: Phase transitions, Preprint
%


\bibitem{nielsen} M. A. Nielsen, I. L. Chuang. \textit{Quantum computation and quantum information}. Cambridge Univ. Press, 2000.
\bibitem{norris} J. R. Norris. \textit{Markov chains}. Cambridge Univ. Press, 1997.
\bibitem{novotny} J. Novotn\'y, G. Alber, I. Jex. Asymptotic evolution of random unitary operations. \textit{Cent. Eur. J. Phys.} {\bf 8}(2010), 1001-1014.


\bibitem{Or} R. Orus, A practical introduction of tensor networks: matrix
product states and projected entangled pair states, \textit{Ann of
Physics} {\bf 349} (2014) 117-158.




\bibitem{portugal} R. Portugal. \textit{Quantum walks and search algorithms}. Springer, 2013.


%
\bibitem{[RoOs96]}
S.  Rommer,  S.  Ostlund,
A class of ansatz wave functions for 1D spin systems and their relation to DMRG,
\textit{Phys. Rev. B} {\bf 55} (1997) 2164.

\bibitem{SP19} I. Sinayskiy , F. Petruccione, Open quantum walks, \textit{ Eur. Phys. J. Spec. Top.} {\bf 227}(2019), 1869–1883.

\bibitem{Sou22} Souissi A., On Stopping Rules for Tree-indexed Quantum Markov chains, preprint (2022)





\end{thebibliography}
\end{document}